\documentclass[runningheads]{llncs}

\usepackage{marvosym}
\usepackage{amsmath,amssymb, amsthm}
\usepackage[normalem]{ulem}
\usepackage{bm}
\usepackage{tikz}
\usepackage{hyperref}
\usepackage{ltl}
\usepackage{graphicx}

\usepackage[strings]{underscore}	 %
\usepackage{pifont}
\usepackage{cleveref}
\usepackage{xspace}
\usepackage[ruled,vlined]{algorithm2e}
\usepackage{mathtools}
\usetikzlibrary{arrows,automata,positioning,shapes,backgrounds,patterns,calc, fit}
\usepackage{subcaption}
\usepackage{xcolor}
\usepackage{multirow}
\usepackage{cite}
\usepackage{orcidlink}
\usepackage{listings}
\usepackage{comment}
\usepackage{booktabs}
\usepackage{wrapfig}

\newsavebox{\measurebox}

\definecolor{dkcyan}{rgb}{0.1, 0.3, 0.3}
\definecolor{dkgreen}{rgb}{0,0.3,0}
\definecolor{olive}{rgb}{0.5, 0.5, 0.0}
\definecolor{dkblue}{rgb}{0,0.1,0.5}

\definecolor{col:ln}{rgb}  {0.1, 0.1, 0.7}
\definecolor{col:str}{rgb} {0.8, 0.0, 0.0}
\definecolor{col:db}{rgb}  {0.9, 0.5, 0.0}
\definecolor{col:ours}{rgb}{0.0, 0.7, 0.0}

\definecolor{lightgreen}{RGB}{170, 255, 220}
\definecolor{darkbrown}{RGB}{121,37,0}

\colorlet{listing-comment}{gray}
\colorlet{operator-color}{darkbrown}
\colorlet{comment-color}{black!50}

\lstdefinelanguage{custom-lang}{
	keywords={let, in, match, with, when, if, then, else, elif, for, to, do, rec, return, new, not, and, while},
	keywordstyle=[1]\color{dkblue}\bfseries,
	morekeywords=[2]{append},
	keywordstyle=[2]\color{dkgreen},
	morekeywords=[3]{synthesize, toPrefixDisting, toAssumeSpec, toGuaranteeSpec, informationClasses, compositionalSynthesis, identicalAPs, append, allTraces, toPrefixDisting, toAssumeSpec, toGuaranteeSpec, informationClasses, identicalAPs, toPrefixNFA, toAssumeNBA, toFullInfNBA, add},
	keywordstyle=[3]\color{dkcyan},
	comment=[l][\color{comment-color}]{//},
	literate=%
	{=}{{{\color{operator-color}=}}}1
	{<-}{{{\color{operator-color}$\leftarrow$}}}1
	{|}{{{\color{dkblue}$\mid$}}}1
	{:}{{{\color{dkblue}:}}}1
	{:=}{{{\color{dkblue}:=}}}1
	{@}{ }1
}

\lstdefinestyle{default}{
	escapeinside={(*}{*)},
	basicstyle=,
	columns=fullflexible,
	commentstyle=\sffamily\color{black!50!white},
	framexleftmargin=1em,
	framexrightmargin=1ex,
	keepspaces=true,
	keywordstyle=\color{dkblue},
	mathescape,
	numbers=left,
	numberblanklines=false,
	numbersep=0.5em,
	numberstyle=\relscale{0.75}\color{gray}\ttfamily,
	showstringspaces=true,
	stepnumber=1,
	xleftmargin=1.2em,
}

\lstnewenvironment{mycode}[1][]
{\small
	\lstset{
		style=default, 
		language=custom-lang,
		#1
	}
}
{}

\newcommand\ldot{\mathpunct{.}}

\newcommand{\variables}{\mathcal V}

\newcommand{\vc}{\mathtt{c}_b}
\newcommand{\vin}{\mathtt{b}_{in}}
\newcommand{\vout}{\mathtt{b}_{out}}

\newcommand{\vt}{\mathtt{t}}

\newcommand{\tracedistinguishability}{\tau}
\newcommand{\prefixdistinguishability}{\rho}
\newcommand{\informationflowassumption}{\mathcal{I}}
\newcommand{\prefixinformationflowassumption}{\mathcal{P}}
\newcommand{\guaranteespecification}{\mathbb{G}}
\newcommand{\trace}{\pi}

\newcommand{\informationclass}{c}
\newcommand{\informationclasses}{\mathcal{C}}
\newcommand{\informationclassspecifcation}{\mathbb{C}}
\newcommand{\assumespecification}{\mathbb{A}}
\newcommand{\hyperimplementation}{\mathcal{H}}
\newcommand{\transitionsystem}{T}
\newcommand{\fullinformationspec}{\mathbb{F}}
\newcommand{\informationclassspecification}{\mathbb{I}}

\theoremstyle{plain}
\newtheorem*{theorem*}{Theorem}
\newtheorem*{lemma*}{Lemma}

\newtheorem{construction}{Construction}

\usepackage{hyperref} %
\usepackage[firstpage]{draftwatermark} %

\SetWatermarkAngle{0}
\SetWatermarkAngle{0}
\SetWatermarkText{\raisebox{17.5cm}{%
\hspace{0.1cm}%
\href{https://doi.org/10.5281/zenodo.10948514}{\includegraphics{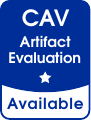}}%
\hspace{11cm}%
\includegraphics{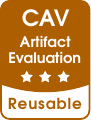}%
}}

\begin{document}

\title{Information Flow Guided Synthesis\newline with Unbounded Communication\thanks{This work was funded by the German Israeli Foundation (GIF) Grant No. I-1513-407./2019,  by DFG grant 389792660 as part of \href{https://perspicuous-computing.science}{TRR~248 -- CPEC}, and by the ERC Grant HYPER (No. 101055412).}}

\titlerunning{Information Flow Guided Synthesis with Unbounded Communication}

\institute{}
\authorrunning{B. Finkbeiner and N. Metzger and Y. Moses}
\author{Bernd Finkbeiner\inst{1}\orcidlink{0000-0002-4280-8441}
\and 
Niklas Metzger\inst{1}\textsuperscript{(\Letter)}\orcidlink{0000-0003-3184-6335}
\and 
Yoram Moses\inst{2}\orcidlink{0000-0001-5549-1781}}

\institute{CISPA Helmholtz Center for Information Security, Saarland, Germany \email{\{finkbeiner, niklas.metzger\}@cispa.de} 
\and
The Andrew and Erna Viterbi Faculty of Electrical and Computer\\ Engineering and the Taub Faculty of Computer Science, Technion, Israel\\
\email{moses@technion.ac.il}
}

\maketitle

\begin{abstract}
Information flow guided synthesis is a compositional approach to the automated construction of distributed systems where the assumptions between the components are captured as information-flow requirements. Information-flow requirements are hyperproperties that ensure that if a component needs to act on certain information that is only available in other components, then this information will be passed to the component. We present a new method for the automatic construction of information flow assumptions from specifications given as temporal safety properties. The new method is the first approach to handle situations where the required amount of information is unbounded. For example, we can analyze communication protocols that transmit a stream of messages in a potentially infinite loop. We show that component implementations can then, in principle, be constructed from the information flow requirements using a synthesis tool for hyperproperties. We additionally present a more practical synthesis technique that constructs the components using efficient methods for standard synthesis from trace properties. We have implemented the technique in the prototype tool \textsc{FlowSy}, which outperforms previous approaches to distributed synthesis on several benchmarks.

\end{abstract}

\section{Introduction}\label{sec:introduction}

More than 65 years after its introduction by Alonzo Church~\cite{Church/57/Applications}, the synthesis of reactive systems, and especially the synthesis of \emph{distributed} reactive systems, is still a most intriguing challenge. In the basic reactive synthesis problem, we translate a specification, given as a formula in a temporal logic, into an implementation that is guaranteed to satisfy the specification for every possible input from the environment. In the synthesis of \emph{distributed systems}~\cite{PnueliR90}, we must find an implementation that consists of multiple components that communicate with each other via shared variables in a given architecture. While the basic synthesis problem is, by now, well-supported with algorithms and tools (cf.~\cite{DBLP:reference/mc/BloemCJ18,DBLP:journals/corr/abs-2206-00251}), and despite a long history of theoretical advances~\cite{Manna+Wolper/84/Synthesis,PnueliR90,kv01,fs05,Madhusudan+Thiagarajan/01/Distributed,Madhusudan+Thiagarajan/02/Decidable}, no practical methods are currently known for the synthesis of distributed systems.

A potentially game-changing idea is to synthesize the systems compositionally, one component at a time~\cite{DBLP:journals/ijfcs/ScheweF07,DBLP:conf/tacas/ChatterjeeH07,DBLP:conf/tacas/BloemCJK15,SafralessCompositionalSynthesis,KuglerS09,CompositionalAlgorithmsforLTLSynthesis, dependency-based}. The key difficulty in automating compositional synthesis is to find assumptions on the behavior of each component that are sufficiently strong so that each component can guarantee the satisfaction of the specification based on the guarantees of the other components, and, at the same time, sufficiently weak, so that the assumptions can actually be realized. In our previous work on \emph{information flow guided synthesis}~\cite{DBLP:conf/cav/FinkbeinerMM22}, we identified situations in which certain components must act on information that these components cannot immediately observe, but must instead obtain from other components. Such situations are formalized as information-flow assumptions, which are hyperproperties that express that the component eventually receives this information. Once the information flow assumptions are known, the synthesis proceeds by constructing the components individually so that they satisfy the information-flow assumptions of the other components provided that their own information-flow assumptions are likewise taken care of. 

Technically, the synthesis algorithm identifies a finite number of sets of infinite sequences of external inputs, so-called \emph{information classes}, such that the component only needs to know the information class, but not the individual input trace. In the first step, the output behavior of the component is fixed based on an abstract input that communicates the information class to the component. This abstract implementation is called a \emph{hyper implementation} because it leaves open how the information is encoded in the actual inputs of the component. Once all components have hyper implementations, the abstract input is then replaced by the actual input by inserting a monitor automaton that derives the information class from the input received by the component. 

This approach has two major limitations. The first is that the information flow requirement only states
that the information will \emph{eventually} be transmitted. This is sufficient for liveness properties where the necessary action can be delayed until the information is received. For safety, however, such a delay may result in a violation of the specification. As a result, the information flow assumptions of~\cite{DBLP:conf/cav/FinkbeinerMM22} are insufficient for handling safety, and \emph{the compositional synthesis approach is thus limited to liveness specifications.}
The second limitation is due to the restriction to a \emph{finite} number of information classes. As a result, the compositional synthesis approach is only successful if a solution exists that \emph{acts on just a finite amount of information}. 
The two limitations severely reduce the applicability of the synthesis method. Most specifications contain a combination of safety and liveness properties (cf.~\cite{DBLP:journals/corr/abs-2206-00251}). While it is possible to effectively approximate liveness properties through \emph{bounded} liveness properties (cf.~\cite{Finkbeiner+Schewe/13/Bounded}), which are safety properties, the converse is not true.  
Likewise, most distributed systems of interest are reactive in the sense that they maintain an ongoing interaction with the external environment. As a result, they do not conform to the limitation that they only act on a finite amount of information. For example, a communication protocol receives a new piece of information in each message and is hence required to transmit an unbounded amount of information. 

In this paper, we overcome both limitations with a new method for information flow guided synthesis that handles both safety properties and specifications of tasks that require the communication of an unbounded amount of information. In order to reason about safety, we consider  \emph{finite} prefixes of external inputs rather than infinite sequences. The key idea is to collect sets of finite sequences \emph{of the same length} into information classes. Such an information class refers to a specific point in time (corresponding to the length of its traces) and identifies the information that is needed at this point in time to avoid a violation of the safety property.
We then only require that the number of information classes is finite \emph{at each point in time}, while the \emph{total} number of information classes over the infinitely many prefixes of an execution may well be infinite. This allows us to handle situations where again and again some information must be transmitted in a potentially infinite loop.

\section{Running Example: Sequence Transmission}\label{sec:running:example}
\begin{figure}[t]

\begin{minipage}{0.5\textwidth}
\begin{subfigure}[b]{\textwidth}
\resizebox{.8\textwidth}{!}{
\tikzstyle{state}=[draw, circle, fill=none, minimum width=0.7cm, 
minimum height = 0.7cm,
align=center, thick]

\begin{tikzpicture}[->,>=stealth',shorten >= 1pt,auto]

\node[state] (p0){%
    $~q_0$
};
\node[state, accepting] (p1) [right = 2 of p0]{%
    $~q_1$
};

\node[state, draw=none] (left)[left = 0.7 of p0]{%
};

\node[state, draw=none] (format)[left = 1.3 of p0]{%
};

\path (left) edge[thick] (p0)
(p0) edge[loop above, thick] node[above ] {%
        $\mathtt{b}_{in_\pi} \leftrightarrow \mathtt{b}_{in_{\pi'}}$
      } (p0)
(p0) edge[thick] node[above] {%
        $\mathtt{b}_{in_\pi} \nleftrightarrow \mathtt{b}_{in_{\pi'}}$
      } (p1)
      ;
\end{tikzpicture}
}
\caption{The distinguishability NFA.}
\label{fig:sequencetransmission:disting}
\end{subfigure}

\begin{subfigure}[b]{\textwidth}
\resizebox{.8\textwidth}{!}{
\tikzstyle{state}=[draw, circle, fill=none, minimum width=0.7cm, 
minimum height = 0.7cm,
align=center, thick]

\begin{tikzpicture}[->,>=stealth',shorten >= 1pt,auto]

\node[state] (p0){%
    $~q_0$
};
\node[state, accepting] (p1) [right = 2 of p0]{%
    $~q_1$
};
\node[state, draw=none] (left)[left = 0.7 of p0]{%
};

\node[state, draw=none] (format)[left = 1.3 of p0]{%
};

\path (left) edge[thick] (p0)
(p0) edge[loop above, thick] node[above] {%
        $\ast$
      } (p0)
(p0) edge[thick] node[above] {%
        $\neg\vin$
      } (p1)
      ;
\end{tikzpicture}
}
\caption{An information class NFA.}
\label{fig:information:class}
\end{subfigure}
\end{minipage}
\begin{minipage}{0.47\textwidth}
\begin{subfigure}[b]{\textwidth}
\resizebox{\textwidth}{!}{

\tikzstyle{state}=[draw, circle, fill=none, minimum width=1.15cm, 
minimum height = 1.15cm, align=center, thick]

\begin{tikzpicture}[->,>=stealth',shorten >= 1pt,auto]

\node[state] (p0){%
    $~\vout$
};

\node (left) [left = 0.5 of p0]{};

\node[state] (p1) [above right = .6 and 2 of p0] {%
  $~\vout$
 };
\node[state] (p2) [below right= .6 and 2 of p0]{
    $\neg \vout$
    };

\path (left) edge[thick] (p0)
      (p0) edge[loop above=20, thick] node[left, xshift=-2pt] {%
        $\neg (\informationclass \vee \informationclass')$
      } (p0)
      (p1) edge[loop right, thick] node[right] {%
        $\informationclass $
      } (p1)
      (p2) edge[loop right, thick] node[right] {%
        $\informationclass' $
      } (p2)
      (p0) edge[thick, bend right=20] node[above] {%
        $\informationclass$
      } (p1)
      (p1) edge[thick, bend right=20] node[above, yshift=6pt, xshift=0pt ] {%
        $\neg (\informationclass \vee \informationclass')$
      } (p0)
      (p0) edge[thick, bend left=20] node [below,  align=center] {%
      $ \informationclass'$
      }(p2)
      (p2) edge[thick, bend left=20] node [below left, xshift=2pt ,  align=center] {%
      $ \neg (\informationclass \vee \informationclass')$
      }(p0)
      (p1) edge[thick, bend right=20] node [left, align=center] {%
      $ \informationclass'$
      }(p2)
      (p2) edge[thick, bend right=20] node [ right, align=center] {%
      $ \informationclass$
      }(p1)
       (p0) edge[loop left, thick, draw=none] node [left] {%
       $\phantom{\LTLtrue}$}
       (p0) %
      ;
\end{tikzpicture}
}
\caption{The safety hyper implementation.}
\label{fig:hyper:receiver}
\end{subfigure}
\end{minipage}
\label{fig:sequencetransmission}
\caption{The prefix distinguishability of the sequence transmission protocol as NFA in (a). The NFA representing the information class for output $\vout$ is shown in (b), where (c) is a hyper implementation of the receiver on information classes.}
\vspace{-7pt}
\end{figure}
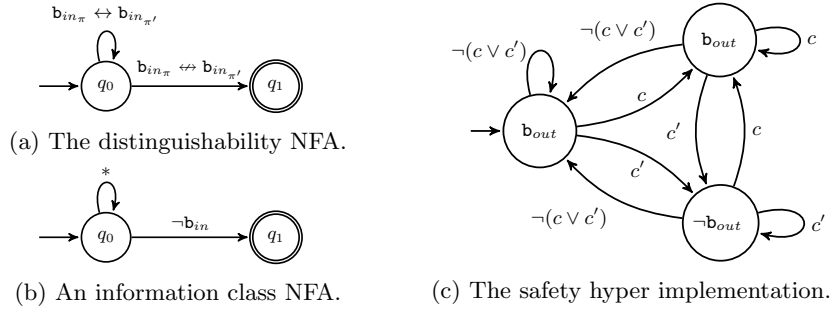

Our running example is a distributed system that implements sequence transmission.
The system consists of two components, the \emph{transmitter} $t$ and the \emph{receiver}~$r$. At every time step, the transmitter observes the current input bit $\vin$ from the external environment, the transmitter can communicate via $\vc$ with the receiver, and the receiver controls the output $\vout$. 
To implement a sequence transmission protocol, the receiver must output the value of the input bit one time step after it is received by the transmitter.
We can state this specification using the LTL formula $\LTLglobally (\vin \leftrightarrow \LTLnext \vout)$ for the receiver, and assume the transmitter specification to be \emph{true}.
In this example, compositional synthesis is only possible with assumptions about the communication between the components. 
We utilize an \emph{information-flow assumption} for compositional synthesis specified in the HyperLTL formula
$\forall \pi \forall \pi'. (\mathtt{c}_{b_{\pi}} \leftrightarrow \mathtt{c}_{b_{\pi'}}) \LTLuntil (\mathtt{b}_{in_{\pi}} \nleftrightarrow \mathtt{b}_{in_{\pi'}} \wedge \mathtt{c}_{b_\pi} \nleftrightarrow \mathtt{c}_{b_{\pi'}})$.
The formula states that on any pair of traces $\pi$ and $\pi'$ of an implementation, the communication bit $\vc$ on both traces must be equivalent until there is a difference on the input bit $\vin$ as well as a difference on the communication bit $\vc$. 
This implies that whenever the receiver must distinguish two input traces, it will observe a difference on its local inputs, namely $\vin$.
A nondeterministic finite automaton (NFA) accepting all finite traces that must be distinguished at the same time point is depicted in \Cref{fig:sequencetransmission:disting}.
In the course of this paper, we show that, for safety properties, the distinguishability requirement yields an information-flow assumption specified over finite traces. 
Based on the assumption, we heuristically build information classes over finite traces, such that all finite traces in the same class do not need to be distinguished.
\Cref{fig:information:class} shows an NFA for one of the two information classes. 
It accepts all finite traces that have $\neg \vin$ in the last step. 
On all these traces, the output $\neg \vout$ is correct.
For this example, there is only one other information class, namely the finite traces with $\vin$ in the last step.
We use the information classes to synthesize a \emph{hyper implementation} for the receiver, depicted in~\Cref{fig:hyper:receiver}.
A hyper implementation receives the current information classes, which are $\informationclass$ and $\informationclass'$ on the transitions, as input, and outputs the local outputs of the component. 
Whenever $\informationclass$ is the input, the correct output for all traces in $\informationclass$ must be set by the receiver. 
Note that, in this example, $\informationclass$ and $\informationclass'$ cannot occur together as there is no common output for $\vin \wedge \neg \vin$.
The hyper implementation is correct for all transmitter implementations.
After synthesizing both hyper implementations, for the transmitter and the receiver, we compose and decompose them to obtain local implementations.
Throughout this paper, we first define the prefix distinguishability and prefix information-flow assumption.
We then build assume and guarantee specifications that, based on the information classes, guarantee the correctness of the hyper implementations, and finally, we show how to construct the local solutions to complete the synthesis procedure.

\section{Preliminaries}

\paragraph{Architectures.} In this paper, we consider distributed architectures with two components: $p$ and $q$. Such architectures are given as tuple $(I_p,I_q,O_p,O_q,O_e)$, where $I_p,I_q,O_p,O_q,$ and $O_e$ are all subsets of the set $\variables$ of boolean variables. $O_p$ and $O_q$ are the sets of \emph{output variables} of $p$ and $q$. We denote by~$O_e$ the output variables of the uncontrollable external environment. We refer to $O_e$ also as the \emph{external inputs} of the system.
$O_p, O_q$ and~$O_e$ form a partition of $\variables$.
Finally, $I_p$ and $I_q$ are the \emph{input variables} of components $p$ and~$q$, respectively. 
The inputs and outputs are disjoint, i.e., $I_p \cap O_p = \emptyset$ and $I_q \cap O_q = \emptyset$. Each of the inputs $I_p$ and~$I_q$ of the components is either an output of the environment or an output of the other component, i.e., $I_p \subseteq O_q \cup O_e$ and $I_q \subseteq O_p \cup O_e$. 
For a set $V \subseteq \variables$, every subset $V' \subseteq V$ defines a \emph{valuation} of $V$, where
the variables in $V'$ have value $\mathit{true}$ and the variables in $V \setminus V'$ have value $\mathit{false}$.

\paragraph{Implementations.} For a set of atomic propositions $AP$ divided into inputs $I$ and outputs $O$, with $I \cap O = \emptyset$,  a $2^{O}$-labeled $2^{I}$-transition system is a 4-tuple $(T, t_0, \tau, o)$, where $T$ is a set of states, $t_0\in T$ is an initial state, $\tau: T \times 2^I\rightarrow T$ is a transition function, and $o: T \rightarrow 2^O$ is a labeling function.
An implementation of an architecture $(I_p,I_q,O_p,O_q,O_e)$ is a pair $(\transitionsystem_p, \transitionsystem_q)$, consisting of $\transitionsystem_p$, a $ 2^{O_p}$-labeled $2^{I_p}$ transition system $\transitionsystem_p$, and 
$\transitionsystem_q$,  a $2^{O_q}$-labeled $2^{I_q}$ transition system $\transitionsystem_q$.
The \emph{composition}  $\transitionsystem = \transitionsystem_p || \transitionsystem_q$ of the two transition systems $(T^p, t^p_0, \tau^p, o^p)$ and $(T^q, t^q_0, \tau^q, o^q)$ is the $2^{O_p \cup O_q}$-labeled $2^{O_e}$-transition system $(T, t_0, \tau, o)$, where $T = T^p\times T^q, t_0 = (t^p_0, t^q_0), \tau((t^p, t^q), x) = (\tau^p(t^p, (x \cup o^q(t^q))\cap I_p), \tau^q(t^q, (x \cup o^p(t^p))\cap I_q))$, $o (t^p, t^q) = o^p(t^p) \cup o^q(t^q)$, where $x \in 2^{O_e}$.

\paragraph{Specifications.}
The specifications are defined over the variables $\variables$.
For a set $V \subseteq \variables$ of variables, a \emph{trace} over $V$ is an infinite sequence $x_0x_1x_2\ldots \in (2^{V})^\omega$ of valuations of $V$.
A \emph{specification} over~$\variables$  is a set $\varphi \subseteq (2^{\variables})^\omega$ of traces over $\variables$.
Two traces over disjoint sets $V,V'\subset\variables$ can be \emph{combined} by forming the union of their valuations at each position, i.e., 
$x_0x_1x_2\ldots \sqcup y_0y_1y_2\ldots = (x_0 \cup y_0)(x_1\cup y_1)(x_2\cup y_2) \ldots$. Likewise, the \emph{projection} of a trace onto
a set of variables $V' \subseteq \variables$ is formed by intersecting the valuations with $V'$ at each position:
$x_0x_1x_2\ldots \downarrow_{V'} = (x_0 \cap V')(x_1 \cap V')(x_2 \cap V')\ldots$.
For a trace $\pi$ we use $\pi[n]$ to access the set on $\pi$ at time step $n$, and $\pi[n\ldots m]$ for the interval of $\pi$ from index $n$ to $m$. 
Our specification language is linear-time temporal logic (LTL)~\cite{LTL} with the set $\variables$ of variables serving as the atomic propositions.
We use the usual Boolean operations, the temporal operators Next $\LTLnext$, Until $\LTLuntil$, Globally  $\LTLglobally$, and Eventually $\LTLeventually$, and the semantic evaluation of (finite) traces $\pi$ with $\pi \vDash \varphi$.
LTL formulas can be represented by nondeterministic B\"{u}chi automata (NBAs) with an exponential blow-up.
A finite trace $\pi \in (2^{\variables})^*$ is a bad prefix of an LTL formula $\varphi$ if $\pi \nvDash \varphi$ and $\pi\cdot\pi' \nvDash \varphi$ for all $\pi' \in (2^{\variables})^\omega$.
An LTL formula is a \emph{safety} formula if every violation has a bad prefix.
Specifications over architectures are conjunctions  $\varphi_p \wedge \varphi_q$ of two LTL formulas, where $\varphi_p$ is defined over $O_p \cup O_e$, i.e., $\varphi_p$ relates outputs of the component $p$ to the outputs of the environment, and $\varphi_q$ is defined over $O_q \cup O_e$. We call these specifications the \emph{local} specifications of the component. 
An initial run $T(i_0, i_1, \ldots) = t_0t_1\ldots \in T^\omega$ for an infinite sequence of inputs $i_0, i_1 \ldots \in 2^{O_e}$ is an infinite sequence of states produced by the transition function such that $t_i = \tau(t_{i-1}, i_{i-1})$ for all $i \in \mathbb{N}$ and $t_0$ is the initial state.
The set of traces $\mathit{Traces}(T)$ of an implementation $T=(T^p, T^q)$ is then defined as all $(o(t_0)\cup i_0)(o(t_1)\cup i_1)\ldots \in (2^\variables)^\omega$ where $
T(i_oi_1\ldots) = t_0t_1\ldots $ for some $i_oi_1i_2\ldots \in (2^{O_e})^\omega$.
An implementation \emph{satisfies} a specification $\varphi$ if the traces of the implementation are contained in the specification, i.e., $\mathit{Traces}(T^p, T^q) \subseteq \varphi$.
Given an architecture and a specification~$\varphi$, the synthesis problem is to find an implementation~$T = (T_p, T_q)$  that satisfies~$\varphi$.
We say that a specification~$\varphi$ is \emph{realizable} in a given architecture if such an implementation exists, and \emph{unrealizable} if not. 

\paragraph{Automata.}
A non-deterministic automaton $\mathcal{A}$ is a tuple  $(\Sigma, Q, q_o, \delta, F)$ where $\Sigma$ is the input alphabet, $Q$ is a set of states, $q_o$ is the initial state, $\delta : Q \times \Sigma \rightarrow 2^{Q}$ is a transition function, and $F$ is a set of accepting states. 
For an input word $\sigma_0\sigma_1\ldots\sigma_k \in \Sigma^k$, a finite word automaton (NFA) $\mathcal{F}$ accepts a finite run $q_0q_1\ldots  q_k \in Q^k$ where $q_i \in \delta(q_{i-1}, \sigma_{i-1})$, if $q_k \in F$.
A Büchi automaton (NBA) $\mathcal{A}$ accepts all infinite runs $q_0q_1\ldots \in Q^{\omega}$ that visit states in $F$ infinitely often.
An automaton is deterministic if the transition function $\delta$ is injective.
The language of an automaton $\mathcal{A}$  is the set of its accepting runs, and is denoted by  $\mathcal{L}(\mathcal{A})$.

\paragraph{Hyperproperties.}
Information-flow assumptions are hyperproperties. 
A \emph{hyperproperty over~$\variables$} is a set $H \subseteq 2^{(2^{\variables})^\omega}$ of sets of traces over $\variables$~\cite{ClarksonS10}. An implementation $(T_p, T_q)$ satisfies the hyperproperty $H$ iff the set of its traces is an element of $H$, i.e., $\mathit{Traces}(T_p, T_q) \in H$. 
A convenient specification language for hyperproperties is the temporal logic HyperLTL~\cite{HyperLTL}, which extends LTL with trace quantification, i.e., $\forall \pi. \varphi$ and $\exists \pi. \varphi$. 
In HyperLTL, atomic propositions are indexed by a trace variables, which make expressing properties like 
``\emph{$\psi$ must hold on all traces}''possible, expressed by $\forall \pi.~\psi\,$. 
Dually, one can express that ``\textit{there exists a trace on which $\psi$ holds}'', denoted by $\exists \pi.~\psi\,$. 
Sometimes, a hyperproperty can be expressed as a binary relation on traces. 
A relation $R \subseteq (2^{\variables})^\omega \times (2^{\variables})^\omega$ of pairs of traces defines the hyperproperty $H$, where a set $T$ of traces is an element of $H$ iff for all pairs $\pi, \pi' \in T$ of traces in $T$ it holds that $(\pi, \pi') \in R$.
We call a hyperproperty defined in this way a \emph{2-hyperproperty}.
In HyperLTL, 2-hyperproperties are expressed as formulas with two universal quantifiers and no existential quantifiers.
A 2-hyperproperty can equivalently be represented as a set of infinite sequences over the product alphabet $\variables^2$: we can represent a given 2-hyperproperty $R \subseteq \variables^\omega \times \variables^\omega$, by $R' = \{ (\sigma_0, \sigma_0') (\sigma_1,\sigma_1') \ldots \mid (\sigma_0\sigma_1\ldots, \sigma_0'\sigma_1'\ldots) \in R \}$. This representation is convenient for the use of automata to recognize 2-hyperproperties.

\section{Prefix Information Flow}\label{sec:safety:information:flow}
As argued in~\cite{DBLP:conf/cav/FinkbeinerMM22}, identifying information flow between the components is crucial for distributed synthesis, because the specification may require a component's actions to depend on external inputs that are not directly observable by the component.
To react to the external inputs correctly, at least the relevant information must be transferred to the component.
The fundamental concept to identify when a component requires information transfer is captured by a \emph{distinguishability} relation on sequences of environment outputs.
We recall the definition of distinguishability for a component $p$ from~\cite{DBLP:conf/cav/FinkbeinerMM22}:

\begin{definition}[Trace distinguishability~\cite{DBLP:conf/cav/FinkbeinerMM22}]\label{def:distinguishability:relation}
Let $\varphi_p$ be an LTL specification of ~$p$. The corresponding \emph{trace distinguishability relation} is defined as 
\vspace{-6pt}
\begin{align*}
		\tracedistinguishability_{p} = \{ (\pi_e, \pi_e') \in &(2^{O_e})^\omega \times (2^{O_e})^\omega \mid\\
		 &\forall \pi_p \in (2^{O_p})^\omega\ldot \pi_e \sqcup \pi_p  \nvDash \varphi_p \text{ or } \pi_e' \sqcup \pi_p \nvDash \varphi_p\}
\end{align*}                
\end{definition}

The trace distinguishability relation is defined w.r.t.\  pairs of infinite traces, where each trace records all outputs of the environment, building up all the information that is presented to the system.
Two traces are related iff there exists no infinite trace of $p$'s outputs that satisfies the specification for both (environment) input traces.
For example, the traces in the sequence transmission protocol are related by $\tracedistinguishability_{r}$ if they differ on $\vin$ at least once. 
We now turn the distinguishability relation into an assumption for the component.
On traces related by $\tracedistinguishability_p$, the component must observe a difference in its \emph{local} inputs, namely the set $I_p$.
The relation itself only considers infinite traces over all variables that are \emph{not outputs} of the single component, independent of the architecture.
Therefore, the information-flow assumption (IFA) built from the distinguishability relation enforces that on all related (environment input) traces, there is a difference on the component's input:
\begin{definition}[Trace information-flow assumption~\cite{DBLP:conf/cav/FinkbeinerMM22}]\label{def:trace:information:flow:assumption}
Let $\tracedistinguishability_p$ be the trace distinguishability relation for $p$.
The \emph{information flow assumption} $\informationflowassumption_p$ 
is the 2-hyperproperty defined by the relation 
\[
 R_{\informationflowassumption_p} = \{ (\pi, \pi') \in (2^{\variables})^\omega \times (2^{\variables})^\omega \mid
\text{ if } (\pi {\downarrow_{O_e}}, \pi' {\downarrow_{O_e}}) \in \tracedistinguishability_p \mbox{ then } \pi {\downarrow_{I_p}} \neq \pi' {\downarrow_{I_p}} \}
\]
\end{definition}
The trace information-flow assumption is necessary for a component $p$; every 
implementation of the distributed system will satisfy the information-flow assumption from~\cite{DBLP:conf/cav/FinkbeinerMM22}.
In its generality, this definition specifies that the values of the local inputs to $p$ have to be different \emph{at some time point}, without an explicit or implicit deadline.
This is critical in two ways:
On the one hand, liveness specifications, as in the example $\vin \leftrightarrow \LTLeventually \vout$, will never determine an explicit point in time where the information must be present.
On the other hand, safety specifications always include a fixed deadline for the reaction of the component, which, if the information is not present, cannot be met.
This deadline, however, is not accounted for in the information-flow assumption, and an algorithm cannot rely on availability of the information during synthesis.

In ~\cite{DBLP:conf/cav/FinkbeinerMM22} we solve the liveness issue by introducing a time-bounded information-flow assumption.
The time bound acts as a placeholder for the exact time point of information flow.
The locally synthesized receiver must then be correct for all such possible time points.
Because of the arbitrary deadline, the assumptions cannot suffice to find a solution for a safety specification of the receiver either; they are too weak.
We solve this issue by restricting the attention to safety specifications.
Consider, for example, the safety property $\varphi_r = \LTLglobally(\vin \leftrightarrow \LTLnext \vout)$ of our running example.
To satisfy this property, the receiver $r$ must observe the value of $\vin$ on its local inputs in exactly one time step, otherwise, it cannot react to $\vin$ in time.
With this observation, we can state a stronger distinguishability relation over pairs of \emph{finite} traces.

\begin{definition}[Prefix distinguishability]\label{def:prefix:distinguishability}
Let $\varphi_p$ be the safety specification for component $p$. The \emph{prefix distinguishability relation} is defined as
\begin{align*}
\prefixdistinguishability_{\varphi_p} = \{(\pi,\pi') \in (2^{O_e})^m \times & (2^{O_e})^m  , m \in \mathbb{N} \mid \forall \pi_p \in (2^{O_p})^m. \\
&\pi \sqcup \pi_p \nvDash_m \varphi_p 
\text{~or~} \pi' \sqcup \pi_p \nvDash_m \varphi_p \\
\text{~and~} \forall n \in &\mathbb{N}, n < m.\ \exists \pi_p' \in (2^{O_p})^n.\\
&\pi[0\dots n] \sqcup \pi_p' \vDash_n \varphi_p 
\text{~and~} \pi'[0\dots n] \sqcup \pi_p' \vDash_n \varphi_p
\}
\end{align*}
\end{definition}

The first condition states that, for the two related input traces of length~$m$, the specification is violated for all possible output sequences for~$p$ of the same length.
The second condition enforces that $m$ is the first position at which the trace pair must be distinguished, i.e., for all previous positions of the traces, there exists a common output sequence that satisfies the specification on both traces.
Every violation of a safety specification has a \emph{minimal bad prefix}~\cite{DBLP:conf/cav/dAmorimR05}, and hence every violation that originates in the indistinguishability of two traces is captured by~\Cref{def:prefix:distinguishability}.
For liveness specifications, no two traces are related by this definition:
One can inductively reason that for every $(\pi,\pi') \in \tracedistinguishability_{\varphi_p}$ this pair of traces is not in $\prefixdistinguishability_{\varphi_p}$, 
i.e., $(\pi,\pi') \notin \prefixdistinguishability_{\varphi_p}$, since for every chosen $m$, one can find an output trace of $p$ that violates the formula after time point $m$. 

Prefix distinguishability is the core concept of our synthesis method.  
We now show that we can build an automaton that accepts a pair of finite environment output traces iff they are related.
We say that an automaton $\mathcal{A}$ \emph{recognizes} a relation $R$ if $\mathcal{L}(\mathcal{A}) = R$.

\begin{theorem} \label{thm:automatonforsifa}
For a component $p$ with specification $\varphi_p$, there exists a non-de\-ter\-mi\-nis\-tic finite automaton with a doubly exponential number of states in the length of $\varphi_p$ that recognizes the prefix distinguishability relation $\prefixdistinguishability_{\varphi_p}$.
\end{theorem}
\begin{proof}\label{cons:prefix:distinguishability}
We construct a non-deterministic finite automaton (NFA) $\mathcal{F}$ that accepts precisely all pairs of traces over $(2^{O_e})^m \times (2^{O_e})^m$, where $m \in \mathbb{N}$, that are related by $\prefixdistinguishability_{\varphi_p}$.
Let $\varphi_p'$ be the formula $\varphi_p$ where all atomic propositions $a \in AP$ are renamed to $a'$, 
and let $\variables'$ be a set containing a copy $v'$ of every variable $v\in\variables$.
We build the NBA $\mathcal{B} = \mathcal{A}_{\varphi_p} \times \mathcal{A}_{\varphi_p'}$, where $\mathcal{A}_{\varphi_p}$ and $\mathcal{A}_{\varphi'}$ are constructed with a standard LTL-to-NBA translation respectively, and the operator $\times$ builds the product of two NBAs.
$\mathcal{B}$ now accepts all tuples of traces that each satisfy $\varphi_p$.
Let $\mathcal{C}$ be the NBA that restricts the transition relation of $\mathcal{B}$ s.t. edges are only present if the output variables of $p$ are equal $\bigwedge_{o \in O_p} o \leftrightarrow o'$ holds, enforcing that both traces agree on the output while satisfying the specification.
We now existentially project to the set $O_e \cup O_e'$ to build $\mathcal{D}$, whose alphabet does not contain the component's outputs.
To accept the pairs of traces that do not satisfy the formula, we negate $\mathcal{D}$, denoted by $\bar{\mathcal{D}}$.
In the last step of the construction, we transform the NBA $\bar{\mathcal{D}}$ to an NFA $\mathcal{F}$ using the \emph{emptiness per state} construction of~\cite{DBLP:journals/tosem/BauerLS11}. 
This yields an NFA that accepts the prefix distinguishability relation.
The size of the automaton is doubly exponential in the size of the formula. 
The first exponent stems from the LTL to NBA construction, and the second from negating the automaton $\mathcal{F}$.
\end{proof}
Similar to~\Cref{def:trace:information:flow:assumption}, we now turn the safety distinguishability relation into an information-flow assumption that must be guaranteed by the component that observes the respective environment output. 
The assumptions include specific information-flow deadlines for pairs of traces at which the component must observe the information at the latest.
The information-flow assumption, again, is a 2-hyperproperty enforcing that pairs of traces that are related by the prefix distinguishability relation have an observable difference for the component.

\begin{definition}[Prefix information-flow assumptions]
Let $\prefixdistinguishability_{\varphi_p}$ be the prefix distinguishability relation for $p$.
The corresponding \emph{prefix information flow assumption} $\prefixinformationflowassumption_p$ 
is the 2-hyperproperty defined by the relation 
\begin{align*}
    R_{\prefixinformationflowassumption_{p}} = \{(\pi, \pi') \in (2^{\variables})^\omega \times (2^{\variables})^\omega\mid  
    \text{~if~}& \exists m \in \mathbb{N} \text{~s.t.~} (\pi[0\dots m], \pi'[0\dots m]]) \in \prefixdistinguishability_{\varphi_p} \\
    &\text{~then~} \pi{\downarrow_{I_p}}[0\dots m-1] \neq \pi'{\downarrow_{I_p}}[0\dots m-1]
    \}
\end{align*}
\end{definition}
On all finite trace pairs in the prefix distinguishability relation $\prefixdistinguishability_{\varphi_p}$, there must be a difference on $I_p$ before the deadline $m$.
Restricting the observable difference to happen before the \emph{deadline} $m$ is crucial for the receiving component.
Whereas the prefix distinguishability relation relates finite traces, the prefix information-flow assumption is a hyperproperty over infinite traces. 
Unsurprisingly, every implementation of a distributed system satisfying safety LTL specifications satisfies the corresponding prefix information-flow assumption.

\begin{lemma}\label{lem:pifa:necessary}
The prefix information-flow assumption is necessary for safety LTL specifications.
\end{lemma}
\begin{proof}
Assume that there exists an implementation $(T_p, T_q)$ satisfying the safety LTL specifications $\varphi_p$ and $\varphi_q$ but not $\prefixinformationflowassumption_p$ and $\prefixinformationflowassumption_q$.
Since $\prefixinformationflowassumption_p$ is not satisfied, there exists a pair of traces $\pi, \pi'$ such that $(\pi {\downarrow_{O_e}}[0\ldots m], \pi' {\downarrow_{O_e}}[0 \ldots m]) \in  \prefixdistinguishability_{\varphi_p}$ 
and
$\pi {\downarrow_{I_p}}[0 \ldots m+1] = \pi' {\downarrow_{I_p}}[0 \ldots m+1]$. 
The deterministic system must therefore choose the same output for the timestep $m+1$ since the inputs are the same.
This contradicts the assumption: either $\pi[0 \ldots m+1]$ or $\pi'[0\ldots m+1]$ is a minimal bad prefix since, otherwise, the traces would not be related by the prefix distinguishability relation.
\end{proof}
We are now ready to return to the sequence transmission example. 
The prefix distinguishability automaton for $\LTLglobally (\vin \leftrightarrow \LTLnext \vout)$ is depicted in \Cref{fig:sequencetransmission:disting}.
The automaton accepts a 2-hyperproperty whose alphabet is a pair of valuations of  $\vin$. 
Note that the communication bit from $t$ to $r$ is not restricted by the prefix distinguishability.
The automaton terminates whenever a sequence of inputs must be distinguished.
For example, starting in the initial state, the input words $\vin$ on $\pi$ and $\neg \vin$ on $\pi'$ lead immediately to an accepting state; these finite traces need to be distinguished.
However, if $\vin$ is equivalent on both traces, the automaton stays in the initial non-accepting state.
By abuse of notation, we use $X_p$ for $X_{\varphi_p}$, e.g.,  
{$\prefixdistinguishability_{p}$ for $\prefixdistinguishability_{\varphi_p}$}, if $\varphi_p$ is clear from context.

The automata for the prefix distinguishability and the prefix infor\-mation-flow assumption can be very complex; even if two traces are different at point~$n$, it can be decided at position $n+m$ if the difference of the inputs results in a necessary information flow, and the automaton might need to store the observed difference during all $m$ intermediate steps.
We evaluate the size of the prefix distinguishability automaton empirically in~\Cref{sec:Implementation}.
With the prefix information-flow assumption, we could construct a hyperproperty synthesis problem similar to \cite{DBLP:conf/cav/FinkbeinerMM22}.
In practice, however, synthesis from hyperproperties is largely infeasible, because it hardly scales to more than a few system states~\cite{HyperBosy}.
In the following, we show that this problem can be avoided by reducing the compositional synthesis problem to the much more practical synthesis from {trace properties}.

\section{Unbounded Communication in Distributed Systems}\label{unbounded:communication:in:distributed:systems}
Computing the information flow between the components in a distributed system, as shown in~\Cref{sec:safety:information:flow}, is the first step for compositional synthesis.
In the second and more complex step, the synthesis procedure needs to guarantee (1) that the component that observes the information actually transmits the information, and (2) that the component requiring the information correctly assumes the reception.
We construct an \emph{assume specification}, which ensures that the component correctly assumes the information flow, 
and a \emph{guarantee specification}, which enforces the correct transmission of information. 

\subsection{Receiving Information}\label{sec:compositional:synthesis:with:information:flow:bounds}
A component cannot realize its specification only based on its local observations; it needs to assume that the required information is transmitted during execution.
The prefix information-flow assumption is one class of necessary assumptions, i.e., every transmitter implementation must satisfy it, and the hyperproperty can be assumed without losing potential solutions.
In many cases, this assumption is also sufficient; if the receiver assumes this exact information flow, the local synthesis problem is realizable.
During synthesis, we do not know what actual information the component currently has.
The synthesis procedure only has \emph{partial information} of all environment outputs.
Which information is actually transmitted at which time point is finally decided by the synthesis process of the \emph{transmitter}. 
However, the receiver's implementation must be correct for every possible information in every step.
We, therefore, collect all traces at a position that do not need to be distinguished by component $p$ at time $n$, i.e., there exists a prefix of $p$'s outputs that works on all traces.

\begin{definition}[Prefix information class]\label{def:information:class}
    Let $\prefixdistinguishability_p$ be the prefix distinguishability relation for $p$. The information class of a trace $\pi$ at position $n \in \mathbb{N}$ is the set of traces
    $[\trace]^n_p = (2^{O_e})^n \backslash \{\pi' \in (2^{O_e})^n \mid (\pi, \pi') \in \prefixdistinguishability_p\}$
\end{definition}

We now construct a trace property that, given an information class $\informationclass^n$, enforces that the output by the component is correct for all traces in the information class $\informationclass^n$.
This property specifies exactly one step of outputs, namely $n+1$.
Since we consider safety LTL properties, it is sufficient to incrementally specify the outputs according to the satisfaction of the LTL formula.

\begin{definition}[Information class specification]\label{def:information:class:specification}
    Let $\varphi_p$ be the LTL specification for component $p$, $n \in \mathbb{N}$, and let $\informationclass_n$ be a prefix information class at position $n-1$.
    The information class specification $\informationclassspecifcation_p^n \subseteq (2^{\variables \backslash O_q})^\omega$ is defined as
    \vspace{-5pt}
    \begin{align*}
    \informationclassspecifcation_p^n = \{\trace_e \sqcup \trace_o \mid \trace_e \in (2^{\variables \backslash O_p})^n&, \trace_o \in (2^{O_p})^n\\
     &\text{ s.t. } \forall \pi_e' \in \informationclass_{n-1}. \trace_e'[0\ldots n] \sqcup \trace_o[0\ldots n] \vDash_n \varphi_p\}.
    \end{align*}
\end{definition}
The output traces in $\informationclassspecifcation^n_p$ need to be correct for every environment output trace that is in the information class. 
Here, if an environment output trace is not in the information class, we do not restrict any behavior.
We now introduce a crucial assumption: That the number of information classes over all time steps is \emph{bounded}.
In general, this is not necessary: one can distinguish every trace from every other trace, such that the information classes increase in every time step.
However, if the number of information classes is bounded, we present an effective heuristic for constructing them on the prefix distinguishability assumption in~\Cref{sec:compositional:synthesis:with:information:flow:bounds}.
Each information class $\informationclass$ (which is now \emph{not} parametric in the time point) is then a set of finite traces $(2^{O_e})^\star$, which is exactly the set of traces in each step that do not need to be distinguished by a component.
Consider, for example, the sequence transmission specification $\varphi = \LTLglobally (\vin \leftrightarrow \LTLnext \vout)$.
The information classes w.r.t.~\Cref{def:information:class} are all traces that are equal on the environment outputs up to time-point $n-1$.
This builds infinitely many information flow classes.
It is, however, possible to reduce the information classes to a finite representation.
In our example, it is sufficient to check for the previous position of the traces: all finite traces that are equal at $n-1$ do not need to be distinguished.
This yields two information classes, one for $\vin$ at the previous step and one for $\neg \vin$ at the previous step.
The NFA accepting one of them is depicted in \Cref{fig:information:class}.
With the assumption that we are given a finite set of information classes as subsets of $(2^{O_e})^*$, we are able to build an \emph{assume specification}, which assumes that information classes are received if necessary, and can react to any possible \emph{consistent} sequence of information classes.
The information classes $\informationclasses$ are now part of the alphabet for the input traces and we use $\informationclass$ for refering to a specific information class and as an atomic proposition.

\begin{definition}[Assume specification]\label{def:assume:specification}
Let $\varphi_p$ be the component specification and $\informationclasses$ be the finite set of information classes, where each $\informationclass \in \informationclasses$ is a subset of $(2^{O_e})^*$. The trace property $\assumespecification \subseteq (\informationclasses \cup 2^{O_p})^\omega$ is defined as
\vspace{-3pt}
\begin{align*}
    \assumespecification_p^\informationclasses = \{\pi_\informationclasses \cup \pi_o \mid \pi_\informationclasses \in \informationclasses^\omega, \pi_o \in &(2^{O_p})^\omega, \forall n \in \mathbb{N}.\forall \informationclass \in \pi_\informationclasses[n-1].\\ \forall \pi_e[0 \ldots n-1] \in \informationclass. 
    &\text{ if } \pi_\informationclass \text{ is consistent, then } \pi_e \sqcup  \pi_o[0\ldots n] \vDash_n \varphi_p\},
\end{align*}
where a finite prefix $\pi_\informationclasses \in \informationclasses^n$ is consistent if it holds that for all $0 \leq m < n$, all finite traces in $\pi_e[0 \ldots m]$ have a prefix in $\pi_e[0 \ldots m-1]$.
\end{definition}

The assume specification collects, for a sequence of information classes, all component outputs that are correct for all environment outputs in this information class.
The consistency of input traces specifies the correct reveal of information classes. 
It cannot be the case that a trace that was distinguishable from the current trace in step $n-1$ is indistinguishable in $n$.
Note that a correct transmitter will implement only consistent traces.
Let's assume we are given the information classes $\informationclasses = \{\informationclass, \informationclass'\}$ for the sequence transmission problem, where $\informationclass = (\{\vin\}, \{\neg \vin\})^*\{\vin\}$ and $\informationclass' = (\{\vin\}, \{\neg \vin\})^*\{\neg \vin\}$. 
These classes suffice to implement the receiver: whenever the trace over $\informationclasses^n$ ends in $\informationclass$, the receiver has to respond with $\vout$ and it should respond with $\neg \vout$ whenever the trace ends in $\informationclass'$.
Each information class $\informationclass$ can be split into the information classes $\informationclass_n$ by fixing the length of the traces to $n$.

\begin{lemma}\label{lem:assume:specification}
Let $\informationclasses_p$ be the finite set of information classes for component $p$.
Every implementation satisfying the assume specification $\assumespecification_p^\informationclasses$ also satisfies the information class specification $\informationclassspecifcation^n_p$ for all $n \in \mathbb{N}$ and $\informationclass \in \informationclasses_p$.
\end{lemma}
This lemma follows directly from the definition of the assume specification: It collects all information class specifications for the given set of information classes.
Note that correctness is only specified for the set of information classes, not the information flow assumption.
If the information classes are not total, in the sense that all distinguished traces are in one of the classes, then the receiver is not correct for \emph{all} implementations of the sender.

\subsection{Transmitting Information}
While a transmitter has to satisfy its local specification, it must also guarantee that the information flow that the receiver relies on is transmitted in time. 
In general, this is, again, a hyperproperty synthesis problem:
The combination of the local specification of $q$ and the prefix information-flow assumption of $p$ is the 2-hyperproperty that the implementation of $q$ needs to satisfy.
However, we propose a framework for more involved (incomplete) trace property synthesis algorithms, potentially speeding up the transmitter synthesis significantly.
In contrast to the receiver, the transmitter of information can choose the synthesis strategy; As long as the transmitter satisfies the information-flow assumption, the receiver will assume this implementation as feasible and can react to the information flow correctly during composition.
We specify a class of trace properties s.t. each element specifies a subset of the implementations that satisfy a correct transmitter.

\begin{definition}[Guarantee specification]\label{def:guarantee:specification}
Let $p$ and $q$ be components and $\informationflowassumption{\varphi_p}$ be the IFA for $\varphi_p$.
The set $\mathbb{G}_{\prefixdistinguishability_p} \subseteq (2^{I_q \cup O_q})^\omega$ is a guarantee specification if all $2^{O_q}$-labeled $2^{I_q}$-transition systems that satisfy $\mathbb{G}$ also satisfy $\informationflowassumption_p$.
\end{definition}

The first crucial difference between the guarantee specification and the assume specification in~\Cref{def:assume:specification} is that the transmitter must guarantee a difference on the traces in $\prefixdistinguishability_p$ whereas the receiver can only assume to observe a difference whenever
$\prefixdistinguishability_p$ relates two traces.
Additionally, the guarantee specification can specify a subset of implementations of all possible transmitters.
We show this difference in the following example:
Consider our running example specification $\varphi = \LTLglobally (\vin \leftrightarrow \LTLnext\vout)$. 
One of the (infinitely) many guarantee specifications can be 
the set of traces specified by the LTL formula 
$\LTLglobally (\vin \leftrightarrow \neg\vc)$, which enforces that every $\vin$ is communicated to the receiver by setting $\vc$ to \emph{false}.

It remains to show that we can construct guarantee specifications for prefix information-flow assumptions effectively.
We will highlight two useful guarantee specifications, one that is implemented in our prototype and one that utilizes the information classes.
We begin with the \emph{full-information specification}.
It forces the transmitter to send, if possible, all information and therefore reduces the distributed synthesis problem to monolithic synthesis.
This concept was already presented in~\cite{PnueliR90} where it was called adequate connectivity and later extended by Gastin et al.~\cite{gastin2009distributed}.
\begin{definition}[Full-information specification]\label{def:full:information:specification}
Let $p$ and $q$ be components, and $f: O_e \cap I_q \rightarrow O_q \cap I_p$ be a bijection.
The \emph{full-information specifciation} for $q$ is the trace property
\begin{align*}
\mathbb{F}_p = \{\pi_e \sqcup \pi_o \mid \pi_e \in (2^{O_e \cap I_q})^\omega, \pi_p \in (&2^{O_q \cap I_p})^\omega, \forall v \in (O_e \cap I_q),\\
\text{ either } &\forall n\in \mathbb{N}. v \in \pi_{e}[n] \text{ iff } f(v) \in \pi_{o}[n+1]\\
\text{ or } &\forall n\in \mathbb{N}. v \in \pi_{e}[n] \text{ iff } f(v) \notin \pi_{o}[n+1] \}
\end{align*}
\end{definition}

This specification forces the sender to assign exactly one value of a communication variable to every input variable.
This choice must hold for every point in time and can not be changed, ensuring that every input combination is uniquely represented by the communication variables.
The full-information specification is a guarantee specification for every possible information-flow assumption.
Since \emph{every} input bit is guaranteed to be transmitted, every different input trace can be distinguished, not only the ones required to be distinguished by the prefix distinguishability relation.
The full-information specification is a sufficient condition for realizing the sender; if there is an implementation for satisfying $\fullinformationspec$, then this implementation is a correct sender.
It is not a necessary specification, the sender might be able to encode the inputs to a smaller set of communication variables.
The second guarantee specification is based on the information classes.

\begin{definition}[Information Class Guarantee]\label{def:information:class:guarantee}
Let $\informationclasses_p'$ be the finite set of information classes of $p$ projected to the inputs of $q$, s.t. the information classes $\informationclass \in \informationclasses_p'$ are subsets of $(2^{O_e\cap I_q})^*$. Let furthermore $f: \informationclasses \rightarrow 2^{O_q \cap I_p}$ be a bijection. The information class guarantee $\mathbb{I}_q^\informationclasses \subseteq (2^{(O_e \cap I_q) \cup (O_q\cap I_p))})^\omega$ is defined as
\begin{align*}
    \informationclassspecification_q^\informationclasses = \{\pi_e \cup \pi_o \mid \pi_e \in (2^{O_e \cap I_q})^\omega, &\pi_o \in (2^{O_q \cap I_p})^\omega, \forall n \in \mathbb{N}, \forall \informationclass \in \informationclasses_p'\\ &\text{ if } \pi_e[0 \ldots n] \in \informationclass \text{ then } f(c) \in \pi_o[n+1]\}.
\end{align*}
\end{definition}

The specification tracks, for an environment output trace $\pi_e$, the current information class.
Whenever the finite trace is in an information class $\informationclass$, the transmitter must set the combination of its outputs to the values as specified by the bijection $f$.
The receiver $p$ can therefore observe $\informationclass$ by decoding the outputs of $q$ on $O_q \cap I_p$.
Similar to the assume specification, the correctness of the information class guarantee depends on the information classes:

\begin{lemma}
If a set of information classes $\informationclasses$ is sufficient to synthesize $\varphi_p$ then $\informationclassspecification_q^\informationclasses$ is a guarantee specification for $\varphi_p$.
\end{lemma}
If providing the information classes at every step is not sufficient for synthesis, then either the specification is unrealizable or at least one information class falsely contains two traces that need to be distinguished.
The assume and guarantee specifications in~\Cref{unbounded:communication:in:distributed:systems} build the foundation for synthesizing local components that satisfy the local specification and the information-flow assumption.
In most distributed systems, however, components are not solely receivers nor transmitters, but both simultaneously.
We now define local implementations that are correct w.r.t. information classes, called safety hyper implementations.

\subsection{Safety Hyper Implementations}
Hyper implementations were introduced in~\cite{DBLP:conf/cav/FinkbeinerMM22} specifying local implementations of a distributed system that are correct for all possible implementations of all other components.
The hyper implementations observe all inputs of the environment but are forced to react to them only if necessary, without restricting the possible solution space of other components.
For example, the implementation of the receiver $r$ in the sequence transmission protocol is a $2^{O_r}$-labeled $2^{I_r}$-transition system, but any locally synthesized solution for r must react to inputs only observed by $t$.
We use the \emph{information classes} of~\Cref{def:information:class} to specify and synthesize a different approach to hyper implementations.
Recall that we assume a bounded number of information classes $\informationclasses$.

\begin{definition}[Safety hyper implementation]\label{def:safety:hyper:implementation}
Let $p$ and $q$ be components, $e$ be the environment, and $\informationclasses_p$ be a set of information classes.
A safety hyper implementation $\hyperimplementation_p$ of $p$ is a $2^{O_p}$-labelled $\informationclasses_p \cup 2^{I_p}$-transition system.
\end{definition}
The safety hyper implementation branches over the information classes and the local inputs to $p$ and reacts with local outputs.
The safety hyper implementation of our running example is depicted in~\Cref{fig:hyper:receiver}.
Compared to (non-safety) hyper implementations in~\cite{DBLP:conf/cav/FinkbeinerMM22}, the safety hyper implementations do not contain a special input variable $\vt$ that signalizes the reception of information. 
This deadline is explicitly present in the prefix distinguishability relation and can be computed on the automaton representing the prefix distinguishability relation.
Since we consider safety properties, there always exists a pre-determined time frame between the environment input and the necessary reception of the information - a fact that we utilize heavily during hyper implementation construction.
We now formalize when a safety hyper implementation is correct.
\begin{definition}[Correctness of safety hyper implementation]
Let $p$, $q$, and $e$ be the components of a distributed system and the environment, and $\varphi_p$, $\varphi_q$ be the local specifications. A \emph{safety hyper implementation} $\hyperimplementation_p$ is correct if it implements $\assumespecification_{\varphi_p}$ and some $\guaranteespecification_{\varphi_q}$.
\end{definition}
Correct hyper implementations of $p$ are compatible with all correct implementations of $q$, i.e., all possible sequences of information provided by \emph{some} transmitter, and implement \emph{one} solution to the information-flow assumption of~$q$.
Since assume and guarantee specifications are trace properties, we can synthesize safety hyper implementations with trace property synthesis algorithms once the B\"uchi automata for the specifications are constructed.
\begin{figure}[t]
\begin{center}
\resizebox{0.91\textwidth}{!}{
\centering
\tikzstyle{state}=[draw, rectangle, minimum width=.8cm, 
minimum height = .8cm, rounded corners=1mm, align=center, thick, draw=black!80]

\begin{tikzpicture}[->,>=stealth',shorten >= 1pt,auto]
\node[state, draw = none, minimum width=0cm, 
minimum height = 0cm] (f1) {%
  $\varphi_p$%
};
\node[state, draw = none, minimum width=0cm, 
minimum height = 0cm] (f2) [below = 1.3 of f1] {%
  $\varphi_q$%
  };

\node[state] (p1) [right = 0.7 of f1] {%
  $\prefixdistinguishability_{\varphi_p}, \informationclasses_{\varphi_p}$%
  };

\node[state] (p2) [right = 0.7 of f2] {%
  $\prefixdistinguishability_{\varphi_q}, \informationclasses_{\varphi_q}$%
  };

\node[state] (t1) [right = 1.5 of p1, align=center] {%

${\assumespecification_p}$%
};

\node[state] (t2) [right = 1.5 of p2, align=center] {%

${\assumespecification_q}$%
};

\node[state, draw=none] (t222) [right = -0.1 of t2, align=center] {%

$\cap$%
};

\node[state, draw=none] (t111) [right = -0.1 of t1, align=center] {%
$\cap$%
};

\node[state, fill = white] (t11) [right = -0.1 of t111, align=right] {%

Full Inf\\
${\guaranteespecification_q}$ ~~Inf Class\\
$\cdots$%
};

\node[state] (t22) [right = -0.1 of t222, align=right] {%
Full Inf\\
${\guaranteespecification_q}$ ~~Inf Class\\
$\cdots$%
};

\node (box1) [state,  fit={(t11)(t1)}] {};
\node (box2) [state,  fit={(t22)(t2)}] {};

\node[state] (h1) [right = 1.5 of t11] {%
$\hyperimplementation_p$%
};

\node[state] (h2) [right = 1.5 of t22] {%
$\hyperimplementation_q$%
};

\node[state, draw = none, minimum width=0cm, 
minimum height = 0cm] (s1) [right = 1.5 of h1] {%
$T_p$%
};

\node[state, draw = none, minimum width=0cm, 
minimum height = 0cm] (s2) [right = 1.5 of h2] {%
$T_q$%
};

\path (f1) edge[thick, transform canvas={yshift=0mm}] node[above] {%
      }  (p1);
      
\path (f2) edge[thick, transform canvas={yshift=0mm}] node[above] {%
      }  (p2)   ;   

\path (p1) edge[thick, transform canvas={yshift=3pt}] node[above] {%
      }  (box1.west);
      
\path (p1.south east) edge[thick, transform canvas={yshift=3pt}] node[above] {%
      }  (box2.west);      
\path (p2) edge[thick, transform canvas={yshift=0mm}] node[above] {%
      }  (box2)      
      ;
\path (p2.north east) edge[thick, transform canvas={yshift=0mm}] node[above] {%
      }  (box1.west)      
      ;

\path (box1) edge[thick, transform canvas={yshift=0mm}] node[above] {%
    synth.
      }  (h1);
      
\path (box2) edge[thick, transform canvas={yshift=0mm}] node[above] {%
synth.
      }  (h2);

\path (h1) edge[thick, transform canvas={yshift=0mm}] node[left,xshift = -6pt,  yshift = 1pt] {%
    compose
      }  node [right, yshift=2pt, xshift=1pt]{decompose}(s2);
      
\path (h2) edge[thick, transform canvas={yshift=0mm}] node[right] {%
      }  (s1);
\end{tikzpicture}
}
\vspace{-10pt}
\end{center}
\caption{The steps in the algorithm for compositional synthesis with prefix information flow assumptions.}
\label{fig:compositional:synthesis:overview}
\vspace{-0.3cm}
\end{figure}
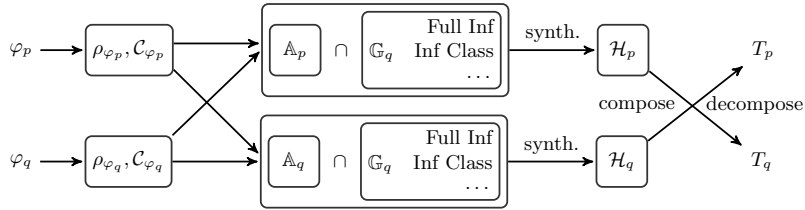
\section{Synthesis with Prefix Information Flow Assumptions}

In this section, we present algorithms for generating assume and guarantee specifications, the synthesis of hyper implementations, and obtaining the solutions for each component.
Combined, this builds our compositional synthesis approach with information-flow assumptions for distributed systems.

\subsection{Automata for Assume and Guarantee Specifications}
The first step in our synthesis approach is to construct the assume specification which builds on a finite set of information classes. 
According to \Cref{def:information:class}, there is, in theory, an unbounded number of information classes.
Our Algorithm~\ref{alg:information:classes} therefore iteratively builds automata that accept, for each prefix length, one information class.
Given the automaton for the prefix distinguishability relation over $\Sigma_\varphi \times \Sigma_{\varphi'}$, the function \lstinline[style=default,language=custom-lang]|identicalAPs| returns an automaton $\mathcal{A}_{id}$ that accepts exactly one input trace over the alphabet $\Sigma_\varphi$ at each time step.
This is achieved by choosing one explicit proposition combination for each edge in the automaton.
\begin{wrapfigure}[12]{R}{0.55\textwidth} %
\vspace{-0.8cm}
\begin{algorithm}[H]
    \caption{Information Classes}
    \label{alg:information:classes}
\begin{mycode}
let informationClasses($\mathcal{A}_{\prefixdistinguishability}$):= 
@@@let $\mathcal{A}_{c}$ = $\mathcal{A}_{\prefixdistinguishability}$
@@@let $\mathcal{A}_{id}$ = identicalAPs($\mathcal{A}_{\prefixdistinguishability}$)
@@@let result = $\emptyset$
@@@while $\mathcal{L}(\mathcal{A}_{id}) \neq\emptyset$ do
@@@@@@result.add(allTraces($\mathcal{A}_{id} \cap \overline{\mathcal{A}_{c}}$))
@@@@@@$\mathcal{A}_{c}$ = $\mathcal{A}_{\prefixdistinguishability} \cap \overline{\mathcal{A}_{id}}$
@@@@@@$\mathcal{A}_{id}$ = identicalAPs($\mathcal{A}_{c}$)
@@@return result

\end{mycode}
\end{algorithm}
\end{wrapfigure}

On this automaton, the function call \lstinline[style=default,language=custom-lang]|allTraces($\mathcal{A}_{id}\cap\overline{\mathcal{A}_c}$)| collects all traces that do not need to be distinguished from $\mathcal{A}_{id}$.
These are the traces in the negation of the prefix distinguishability relation that are related to $\mathcal{A}_{id}$.
The function \lstinline[style=default,language=custom-lang]|allTraces| can be computed by renaming the primed propositions on the edges of the automaton.
This concludes the computation of the first information class.
The algorithm continues by removing $\mathcal{A}_{id}$ from the prefix distinguishability automaton and computing the next information class until the current automaton for the prefix distinguishability relation is empty.

Algorithm~\ref{alg:information:classes} yields, if it terminates, $n$ finite automata $\mathcal{F}_\informationclass$ 
where all traces in each $\mathcal{F}_\informationclass$ do not need to be distinguished.
This implies that there exists a common output combination for each time-step that is correct for each trace in the automaton.
We now show a construction for the assume specification in~\Cref{def:assume:specification}.
\begin{construction}\label{cons:assume:specification}
We first transform the $n$ finite automata $\mathcal{F}_1, \ldots,\mathcal{F}_n$ for the information classes, as obtained as the result of Algorithm~\ref{alg:information:classes}, to the respective information class specification (see~\Cref{def:information:class:specification}).
For each automaton, we build the intersection of the goal automaton $\mathcal{A}_\varphi$ and $\mathcal{F}_i$. 
The resulting automaton accepts all traces in the information class with outputs as specified by $\varphi$.
This yields an NBA $\mathcal{B}$ that only accepts a subset of all input traces.
We lift it to an automaton for the information class specification by unionizing all input and output combinations that do not occur on $\mathcal{F}_i$, which is $\mathcal{A}_{true} \backslash \mathcal{B}$, where $\mathcal{A}_{true}$ is the automaton accepting all input and output combinations.
After performing this construction for all $n$ information class automata, the intersection of all of them accepts the assume specification.
\end{construction}

We use this automaton for the local synthesis of each component. 
The local specification is implicitly satisfied by the hyper implementation of the assume specification since it is encoded in the construction.
We now show how to construct the full information specification in \Cref{def:full:information:specification}.
\begin{construction}\label{con:full:information:specification}
Let $I = I_q\cap O_e$ be the inputs observed by $q$ and $O = O_q \cap I_p$. We assume that $|I| \leq |O|$ since we can only transmit all information if we have at least as many communication variables as environment output variables.
Let $f : I \rightarrow O$ be a bijection that maps input variables to output variables. 
We construct the LTL formula
$ \varphi = \bigwedge_{i \in I} \LTLglobally(i \leftrightarrow \LTLnext f(i)) \vee \LTLglobally (i \leftrightarrow \LTLnext \neg f(i))$.
This formula enforces that, for every $i\in I$, either the value of $i$ is copied to $f(i)$ at every point in time, 
or the negation of~$i$'s value is copied to $f(i)$ at every point.
The corresponding automaton whose language is a full-information specification is $\mathcal{A}_{\varphi}$, obtained by a standard LTL to NBA translation.
\end{construction}

Together with a guarantee specification, the hyper implementation satisfies its own local specification and the guarantee of the other component.

\subsection{Compositional Synthesis}
The last step in the compositional synthesis algorithm is the composition and decomposition of the hyper implementations.
After this process, we obtain the local implementations of the components and therefore the implementation of the distributed system. 
During composition and decomposition, we need to replace the information class variables with the actual locally received input.
The composition collects all environment and component outputs, as well as the information classes for both components.
This includes unreachable states, namely combinations of information classes and environment outputs that are impossible (the finite environment output trace is not in the information class).
We eliminate these states in \Cref{def:enforce:consistency}.
The composition is defined as follows:

\begin{definition}[Composition]\label{def:composition}
Let $p, q$ be components and $\hyperimplementation_p = (T^p, t^p_o, \tau^p,\\ o^p)$ and $\hyperimplementation_q = (T^q, t^q_0, \tau^q, o^q)$ be their respective safety hyper implementations. The composition $\hyperimplementation = \hyperimplementation_p || \hyperimplementation_q$ is a $2^{O_p\cup O_q}$-labeled $2^{O_e}\cup \informationclasses_p \cup \informationclasses_q$-transition system $(T, t_p, \tau, o)$, where $T = T^p \times T^q $, $t_0 = (t^p_0, t^q_0)$, $o ((t^p, t^q)) = o^p(t^p) \cup o^q(t^p)$, and
\begin{align*}
    \tau ((t^p, t^q), x) =~ &\tau^p(t^p, (x \cup o^q(t^q))\cap (I_p \cup \informationclasses_p)),
    \tau^q(t^q, (x \cup o^p(t^p))\cap (I_q \cup \informationclasses_q))
\end{align*}
\end{definition}
The state space is the cross product of the hyper implementations and the labeling function is the union of the local hyper implementation's outputs.
The transition function ensures that the global inputs over $2^{O_e} \cup \informationclasses_p \cup \informationclasses_q$ are separated into the inputs of the respective hyper implementations, namely $\informationclasses_p \cup I_p$ and $\informationclasses_q \cup I_q$.
For every state in the cross-product, the composition branches for every environment output and information class to a local state of a component.
Some of these states are unreachable.
For our running example, the composition includes a transition with $\neg \vin, \informationclass'$, even though the trace with $\neg \vin$ in the last step cannot be in $\informationclass'$.
We now filter states according to consistency.
We consider $\hyperimplementation(x)$ as the hyper implementation $\hyperimplementation$ terminating in $x$.

\begin{definition}[Filter]\label{def:enforce:consistency}
Let $\hyperimplementation = (T, t_0, \tau, o)$ be the composition of the $2^{O_p}$-labeled $2^{\informationclasses_p \cup I_p}$-transition system $\hyperimplementation_p$ and the $2^{O_q}$-labeled $2^{\informationclasses_q \cup I_q}$-transition system $\hyperimplementation_q$.
The consistent composition of 
$\hyperimplementation_p$ and $\hyperimplementation_q$ is 
the hyper implementation $\hyperimplementation'=(T', t'_0, \tau', o')$, with $T' = T $, $t_0' = t_0$, $o' = o$, and
\begin{align*}
    \tau' ((t^p, t^q), x) =~ 
    \left\{
	\begin{array}{ll}
		\tau ((t^p, t^q), x)  & \mbox{ if } \forall \informationclass  \in x. \hyperimplementation(t^p, t^q) \subseteq \mathcal{L}(\mathcal{F}_c) \\
		\emptyset & \mbox{ else} 
	\end{array}
\right.
\end{align*}
\end{definition}

A finite trace $\pi$ of length $n$ over $2^{O_e} \cup \informationclasses_p \cup \informationclasses_{q}$ is impossible to reach if $\informationclass$ is in $\pi[n]$ but  $\pi\downarrow_{O_e}$ is not in the information class represented by $\informationclass$.
Computing if a state is unreachable includes language inclusion of the subsystem terminating in the state and the automaton of the information class. 
However, an algorithm that enforces consistency can monitor the current information class of a state during a forward traversal of the composed hyper implementations. 
In the next and final step, the decomposition then projects the composition to only the \emph{observable} outputs of a component. 
For some input combinations, this yields a set of reachable states, of which we choose one for the decomposition. 
In essence, all these states are viable successors for the current input combination.

\begin{definition}[Decomposition]\label{def:decomposition}
Let $\hyperimplementation = (T, t_p, \tau, o)$ be the consistent composition of the $2^{O_p}$-labeled $2^{\informationclasses_p \cup I_p}$-transition system $\hyperimplementation_p$ and the $2^{O_q}$-labeled $2^{\informationclasses_q \cup I_q}$-transition system $\hyperimplementation_q$.
Furthermore, let $\mathit{min}$ be a function returning the minimal element for a subset of $T$ w.r.t.\ some total ordering over the states of $T$. 
The decomposition $\hyperimplementation|_p$ is a $2^{O_p}$-labelled $2^{I_p}$-transition system $(T^p, t^p_0, \tau^p, o^p)$ where $T^p = T$, $t^p_0 = t_0$, $o^p((t^p, t^q)) = o((t^p, t^q)) \cap O_p$, and
\vspace{-3pt}
\begin{align*}
    \tau^p (t, x) = \mathit{min} \{t'  \mid \exists y \in 2^{(O_e \cup \informationclasses_p )\backslash I_p}. t' = \tau(t, x \cup y) \}
\end{align*}
\end{definition}

The full compositional synthesis algorithm is shown in~\Cref{fig:compositional:synthesis:overview}.
Given the two local specifications, the first step is computing the prefix distinguishability NFAs.
Based on those, the assume specifications and guarantee specifications for both components are constructed and build the inputs to the local synthesis procedures.
Note that the guarantee specification can be any strategy that implements the information flow assumption, e.g., any scheduling paradigm.
After intersecting the two automata, the components must satisfy the assume and the guarantee specification together, which is achieved by trace property synthesis on the intersection of the automata.
The problem is unrealizable if either the prefix information-flow assumption is not sufficient for synthesis (there could be necessary behavioral assumptions), or not all information can be communicated to the receiver.
After composition, consistency, and decomposition, the algorithm terminates with two local implementations that, together, implement a correct distributed system:

\begin{corollary}
Let $p$ and $q$ be components with local specifications $\varphi_p$ and $\varphi_q$. 
The distributed system implementation returned by the algorithm depicted in \Cref{fig:compositional:synthesis:overview} satisfies the local specifications.
\end{corollary}

\section{Experiments}\label{sec:Implementation}

We implemented the compositional synthesis algorithm described so far in our prototype called \textsc{FlowSy}.
The implementation builds on the popular infinite word automaton manipulation tool \textsc{spot}~\cite{duret.22.cav} for translation, conversion, and emptiness checking of NBAs.
\text{FlowSy} implements the support for the finite automata, the construction of prefix distinguishability in Construction~\ref{cons:prefix:distinguishability}, the construction of the information classes in Algorithm~\ref{alg:information:classes}, and building automata for the assume specification in \Cref{cons:assume:specification} and full information specification \Cref{con:full:information:specification}.
The synthesis of the hyper implementations is performed by converting the %
Büchi automata to deterministic parity games and solving them with the solver \textsc{oink}~\cite{DBLP:conf/tacas/Dijk18}.
We report on two research questions, (1) how do the prefix distinguishability automaton and the information classes scale w.r.t.\  formula size and information flow over time and (2) how does \textsc{FlowSy} compare to the existing bounded synthesis approach for distributed system \textsc{HyperBosy} presented in~\cite{HyperBosy}.
Note that, at the time of evaluation, \textsc{HyperBosy} was the only tool for distributed synthesis that we were able to compare against.
A comparison with the existing information flow guided synthesis algorithm with bounded communication in~\cite{DBLP:conf/cav/FinkbeinerMM22} is infeasible since the supported languages of input specifications are disjoint.
All experiments are run on a 2.8 GHz processor with 16 GB RAM, the timeout was 600 seconds, and the results are shown in \Cref{tab:experiments}.
\begin{table}[t]
	\caption{This table summarizes the experimental results. The Benchmark and Parameter columns specify the current instance. The columns $|\varphi|$, $|\prefixdistinguishability|$, and $|\informationclasses|$ give the size of the formula, the number of states in the prefix distinguishability automaton, and the number of information classes, respectively. The last two columns report the running time of \textsc{FlowSy} and \textsc{BosyHyper} in seconds.}\label{tab:experiments}
		\centering
		 \def\arraystretch{1.05}
		 \setlength\tabcolsep{1.7mm}
		\begin{tabular}{lcccccc}
			
			\textbf{Benchmark} & \textbf{Par.}  & $\boldsymbol{|\varphi|}$ & $\boldsymbol{|\prefixdistinguishability|}$ & $\boldsymbol{|\informationclasses|} $ &  \textbf{\textsc{FlowSy}} &  \textbf{\textsc{BosyHyper}}\\
			\toprule
			Delay & 1 & 5 & 4 & 2 & 1.74 & 0.97 \\
                 & 2 & 6 & 8 & 2 & \textbf{1.87} & TO \\
                 & 3 & 7 & 16 & 2 & \textbf{1.84} & TO \\
                 & 4 & 8 & 32 & 2 & \textbf{1.94} & TO \\
                 & 5 & 9 & 64 & 2 & \textbf{2.36} & TO \\
            \midrule
            Sequence Transmission & 1 & 5 & 4 & 2 & 1.83 & \textbf{1.42} \\
			 & 2 & 11 & 6 & 4& \textbf{5.28} & TO \\
                & 3 & 16 & 10 & 8 & \textbf{36.81} & TO \\
            
             \midrule
            Conjunctions & 1 & 5 & 4 & 2 & 3.18 & \textbf{0.92} \\
                         & 2 & 9 & 4 & 4 & \textbf{4.35} & 91.80 \\
                         & 3 & 13 & 4 & 8 & \textbf{9.20} & TO \\
                         & 4 & 17 & 4 & 16 & TO & TO \\
            \midrule
            Disjunctions & 1 & 5 & 4 & 2 & \textbf{3.25} & 6.26 \\
                         & 2 & 9 & 4 & 4 & \textbf{5.63} & 60.08 \\
                         & 3 & 13 & 4 & 8 & \textbf{12.14} & TO \\
                         & 4 & 17 & 4 & 16 & TO & TO \\
\midrule		
  \end{tabular}
\end{table}
\paragraph{Benchmarks.}
The benchmarks scale in 3 different dimensions: the number of independent variables, time-steps in between information reception and corresponding output, and combinatorics over input and output variables.
The first one is independent communication of $n$ input variables in \emph{sequence transmission}.
This parametric version of the running example has $n$ conjuncted subformulas of the form $\LTLglobally (i \leftrightarrow \LTLnext o)$.
For the \emph{delay} benchmark, the number of variables is constant, but the number of time steps between input and output is increased, i.e., the formulas have the form $\LTLglobally (i \leftrightarrow \LTLnext^n o)$.
The last two benchmarks build Boolean combinations over the inputs.
The \emph{conjunctions} benchmark enforces that the conjunctions over the inputs are mirrored in the outputs.
\emph{Disjunctions} is constructed in the same way but with disjunctions in between variables.
Formulas are $\LTLglobally(i_1 \wedge i_1\wedge \ldots \leftrightarrow \LTLnext o_1 \wedge o_2 \wedge \ldots)$ and $\LTLglobally(i_1 \vee i_1\vee \ldots \leftrightarrow \LTLnext o_1 \vee o_2 \vee \ldots)$.

\paragraph{Scaling.}
\textsc{FlowSy} primarily scales in the number of computed information classes. 
Most interestingly, for benchmark \emph{delay}, the number of information classes is constantly 2, even though the size of the prefix distinguishability automaton grows exponentially. 
Independent of the length of the current trace, the automaton for the information class checks that the current position is equal to the position $n$ steps earlier. This can indeed be represented by two information classes.
For synthesizing the conjunction and disjunction benchmarks, the situation is reversed. Even though the prefix distinguishability automaton is constant, the number of information classes grows exponentially in the parameter, collecting all possible combinations of input variables. 
For the sequence transmission benchmark, all reported values scale with the input parameter, which leads to an expected increase of the running time until the timeout at step 4 (not included in~\Cref{tab:experiments}).

\paragraph{Comparison to BosyHyper.}
\textsc{FlowSy} clearly outscales \textsc{BosyHyper}. 
Most interestingly, the delay benchmark shows the almost constant running time for \textsc{FlowSy}. 
Since the number of information classes stays the same, the synthesis of the hyperproperties only scales for transmitting the information. 
\textsc{BosyHyper} must store all values for all $\LTLnext^n$ steps during synthesis, which immediately increases the search space to an infeasible size.
For the benchmarks conjunction and disjunction, one can observe that, although the information classes scale exponentially, the running time of \textsc{FlowSy} is significantly faster than that of \textsc{BosyHyper}, which is already at 91 seconds for parameter 2.
In summary, the compositionality of \textsc{FlowSy} is always beneficial for the synthesis process and it saves on the execution time dramatically when the complex communication in the distributed system can be reduced to a small number of information classes.

\section{Related Work}
Compositional synthesis for monolithic systems, i.e., architectures with one component and the environment, is a well-studied field in reactive synthesis, for example in ~\cite{SafralessCompositionalSynthesis, KuglerS09, CompositionalAlgorithmsforLTLSynthesis, dependency-based} and most recently in~\cite{akshay2024dependent}.
In multi-component systems with partial observation, compositionality has the potential to improve algorithms significantly, for example in reactive controller synthesis~\cite{DBLP:conf/cav/AlurMT16,DBLP:conf/atva/Hecking-Harbusch19}.
Assume-guarantee synthesis adheres to the same synthesis paradigm as our approach: the local components infer assumptions over the other components to achieve the local goals~\cite{DBLP:conf/tacas/ChatterjeeH07,DBLP:conf/tacas/BloemCJK15}.
The assumptions are \emph{trace properties}, restricting the behavior of the components which often is not necessary. 
If the assumptions are not sufficient, i.e., too weak to locally guarantee the specification, the assumption can be iteratively refined~\cite{DBLP:journals/tcad/MajumdarMSZ20}.
Another approach is weakening the acceptance condition to dominance~\cite{CompositionalSynthesisofDistributedSystems} or certificates that specify partial behavior of the components in an iterative fashion~\cite{DBLP:conf/atva/FinkbeinerP21}.
In our previous work on 
information flow guided synthesis~\cite{DBLP:conf/cav/FinkbeinerMM22}, we have
introduced the concept of compositional synthesis with information-flow assumptions. The work presented in the paper overcomes the two major limitations of this original approach, namely the limitation to liveness (or, more precisely, co-safety properties) and the limitation to specifications that can be realized by acting only on a a finite amount of information. 

\section{Conclusion}\label{sec:conclusion}

We have presented a new method for the compositional synthesis of distributed systems from temporal specifications. 
Our method is the first approach to handle situations where the required amount of information is unbounded. While the information-flow assumptions are hyperproperties, we have shown that standard efficient synthesis methods for trace properties can be utilized for the construction of the components. 
In future work, we plan to study the integration of the information-flow assumptions computed by our approach with the assumptions on the functional behavior of the components generated by techniques from behavioral assume-guarantee synthesis~\cite{DBLP:conf/tacas/ChatterjeeH07,DBLP:conf/tacas/BloemCJK15}. Such an integration will allow for the synthesis of systems where the components collaborate both on the distribution and on the processing of the distributed information.

\bibliographystyle{splncs04} 
\bibliography{bibliography}

\end{document}